\documentclass[letterpaper, 10 pt, conference]{ieeeconf}

\IEEEoverridecommandlockouts               

\usepackage{cite}
\usepackage{amsmath,amssymb,amsfonts}
\usepackage{algorithmic}
\usepackage{graphicx}
\usepackage{textcomp}
\usepackage{xcolor}
\def\BibTeX{{\rm B\kern-.05em{\sc i\kern-.025em b}\kern-.08em
    T\kern-.1667em\lower.7ex\hbox{E}\kern-.125emX}}

\usepackage{siunitx}
\usepackage{tikz}
\usetikzlibrary{3d} 
\usetikzlibrary{calc}
\usetikzlibrary{shapes,arrows}
\usetikzlibrary{matrix} 
\tikzset{%
  block/.style    = {draw, thick, rectangle, minimum height = 3em,
    minimum width = 3em},
  sum/.style      = {draw, circle, node distance = 2cm}, 
  input/.style    = {coordinate}, 
  output/.style   = {coordinate} 
}

\usepackage{graphics} 
\usepackage{epsfig} 
\usepackage{subcaption}
\usepackage{booktabs}  
\usepackage{multirow}

\newtheorem{theorem}{Theorem}

\newtheorem{remark}{Remark}

\newenvironment{problem*}{%
  \par\textit{Problem:} 
}{\par} 

\newenvironment{theorem*}{%
  \par\textit{Theorem:} 
}{\par} 

\usepackage[acronym,nonumberlist,nopostdot,nogroupskip]{glossaries}
\newacronym{kf}{KF}{Kalman filter}
\newacronym{ekf}{EKF}{extended Kalman filter}
\newacronym{eqf}{EqF}{equivariant filter}
\newacronym{ltv}{LTV}{linear time-varying}
\newacronym{slam}{SLAM}{simultaneous localization and mapping}
\newacronym{iekf}{IEKF}{invariant extended Kalman filter}
\newacronym{imu}{IMU}{inertial measurement unit}
\newacronym{fov}{FoV}{field of view}
\newacronym{ibvs}{IBVS}{image-based visual servoing}
\newacronym{mle}{MLE}{maximum likelihood estimator}

\newcommand{\mring}[1]{\mathring{#1}}
\newcommand{\T}[0]{\mathsf{T}} 

\newcommand{\mc}[1]{{\ensuremath{\mathcal{#1}}}}  
\newcommand{\mr}[1]{{\ensuremath{\mathrm{#1}}}}   
\newcommand{\mbb}[1]{{\ensuremath{\mathbb{#1}}}}   
\newcommand{\mf}[1]{{\ensuremath{\mathfrak{#1}}}}   

\DeclareMathAlphabet{\mathpzc}{OT1}{pzc}{m}{it} 

\usepackage[textwidth=15mm]{todonotes}

\usepackage{cuted} 
\usepackage{ulem}

\newcommand{\bmat}[1]{\begin{bmatrix}#1\end{bmatrix}}

\allowdisplaybreaks 

\title{\LARGE \bf
An Equivariant Filter for Spacecraft Nonlinear Relative Navigation with Range and Bearing Measurements
}

\author{Gil Serrano$^{1}$, Bruno J. Guerreiro$^{1,2}$, Pedro Lourenço$^{3}$, and Rita Cunha$^{1}$
\thanks{The work of G. Serrano was supported by the PhD Grant from MIT Portugal and Funda\c{c}{\~a}o para a Ci{\^e}ncia e a Tecnologia
(FCT) [DOI: 10.54499/PRT/BD/154275/2022]. This work was also supported by FCT, Portugal through LARSyS [DOI: 10.54499/LA/P/0083/2020 and 10.54499/UID/50009/2025]}
\thanks{$^{1}$The authors are with the Institute for Systems and Robotics, Laboratory of Robotics and Engineering Systems, Instituto Superior Técnico, University of Lisbon, Portugal.
        ({\tt\small gil.serrano@tecnico.ulisboa.pt, rita@isr.tecnico.ulisboa.pt})}%
\thanks{$^{2}$B. J. Guerreiro is also with CTS/Uninova and LASI, School of Science and Technology, NOVA University Lisbon, Caparica, Portugal.
        ({\tt\small bj.guerreiro@fct.unl.pt})}%
\thanks{$^{3}$P. Lourenço is with the GNC Division, Flight Segment and Robotics, GMV, Lisbon, Portugal.
        ({\tt\small palourenco@gmv.com})}%
}

\begin{document}

\maketitle
\thispagestyle{empty}
\pagestyle{empty}


\begin{abstract}

Relative navigation is a fundamental task in space proximity operations and autonomous rendezvous. 
The nonlinear relative orbital dynamics are equivariant and admit a semidirect product Lie group symmetry. 
This property is used to design an Equivariant Filter (EqF). 
The filter estimates the relative position and velocity between a chaser spacecraft and a target by lifting the filter dynamics to the group, while respecting the underlying geometry of the problem.
Simulations demonstrate the filter’s performance and effectiveness.
\end{abstract}


\section{INTRODUCTION}
\label{sec:introduction}

Accurate and robust relative state estimation is an essential requirement for autonomous spacecraft proximity operations. In recent years, rendezvous and capture operations for future On-Orbit Servicing (OOS) and Active Debris Removal (ADR) missions have attracted significant interest, both in industry and academia, due to the increase in satellite deployments \cite{biesbroek_edeorbit_2017,hatty_viability_2022}. Although proximity missions have historically involved human operators, there has been a push for fully autonomous methods, which entail the development of more complex and precise solutions.

In an OOS/ADR scenario, a chaser or servicer spacecraft needs to approach a possibly uncooperative target, for which it requires an estimate of the relative state between both. Even when applying simplifying assumptions, such as no perturbations or a circular target orbit, the relative orbital motion is governed by nonlinear equations. Typically, control and navigation methods for proximity operations in LEO use the Clohessy-Wiltshire (CW) equations \cite{clohessy_rendezvous_1960} to model the system. The CW equations are derived by linearizing the relative dynamics in the target's local-vertical-local-horizontal (LVLH) frame, assuming that the target's orbit is circular, the Earth is spherically symmetric, and the distance between the spacecraft is small with respect to the orbit radius. For elliptical orbits, the linear time-varying Tschauner-Hempel equations are commonly used \cite{alfriend_formation_2009, fehse_rendezvous_2003}.

Despite the widespread use of these linear models \cite{grzymisch_optimal_2015, luo_proximity_2016, napolano_multisensor_2023, nocerino_navigation_2026}, the nonlinear relative orbital dynamics exhibit a property called equivariance, which can be advantageous for navigation. The system admits a nontrivial symmetry structure that can be represented using a semidirect product Lie group.
This property allows one to exploit the symmetry of the system, while respecting the underlying geometry, to design robust observers \cite{mahony_observer_2022}.
One such observer is the Equivariant Filter (EqF) \cite{van_goor_equivariant_filter_2023}. 
The EqF works by expressing the observer state in the symmetry group and lifting the dynamics of the system, linearizing the error dynamics on the state manifold around the origin, and applying principles from EKF design \cite{mahony_equivariant_2021, van_goor_equivariant_filter_2023}. As such, the EqF leverages the inherent symmetry and equivariance of the system to obtain more consistent and robust estimates, compared to more traditional filters that may disregard the underlying geometry of the state space.
This method has been successfully applied to several problems of interest in the control and robotics community, such as inertial navigation \cite{fornasier_equivariant_2022}, VIO \cite{van_goor_eqvio_2023}, and SLAM \cite{ge_equivariant_2025}.
In the space domain, it has been applied to the estimation of relative attitude and angular velocity \cite{serrano_equivariant_attitude_2026}.

Despite these developments, the application of equivariant filtering to relative orbital navigation remains largely unexplored. While the group structure we propose for this problem has previously been considered for simpler dynamical systems \cite{ge_discrete_2022}, its use for nonlinear relative orbital dynamics introduces additional structure between the states.
Thus, the main contributions of this work are: (i) the identification of an equivariant structure for nonlinear relative orbital dynamics; (ii) the derivation of the corresponding Equivariant Filter; (iii) the validation of the approach via simulation.

The remainder of this letter is organized as follows: Section\;\ref{sec:problem_statement} describes the problem at hand; Section\;\ref{sec:equivariant_formulation} details the equivariant structure, introducing the Lie group symmetry, proving equivariance of the system and deriving the filter; Section\;\ref{sec:simulations} presents numerical results that validate the approach; Section\;\ref{sec:discussion} examines the method and findings; and, finally, Section\;\ref{sec:conclusions} offers concluding remarks. 


\section{PROBLEM STATEMENT}
\label{sec:problem_statement}

Consider an unactuated target satellite and an actuated chaser satellite orbiting a central body. We denote the target's position and the chaser's position as ${r_{t},\,r_{c}\in\mbb{R}^{3}}$, respectively. The unperturbed Keplerian two-body equations of motion of each spacecraft, with respect to an Earth-centered inertial (ECI) frame $\{\mc{I}\}$, are given by
\begin{gather}
    \ddot{r}_{t} = -\mu \frac{r_{t}}{\|r_{t}\|^{3}} ,
    \label{eq:target_satellite_orbital_dynamics}\\
    \ddot{r}_{c} = -\mu \frac{r_{c}}{\|r_{c}\|^{3}} + u_{c} ,
    \label{eq:chaser_satellite_orbital_dynamics}
\end{gather}
where $\mu$ is the geocentric gravitational constant and ${u_{c}\in\mbb{R}^{3}}$ is the control acceleration applied on the chaser. 

The relative position between the target and the chaser is ${r = r_{c} - r_{t}}$. The nonlinear relative orbital dynamics, expressed in the inertial frame, are given by
\begin{equation}
    \ddot{r} = -\mu \left( \frac{r+r_{t}}{\|r+r_{t}\|^3} - \frac{r_{t}}{\|r_{t}\|^3} \right) + u_{c} .
    \label{eq:relative_orbital_dynamics_inertial}
\end{equation}

If instead one expresses the motion in the target's rotating LVLH frame, $\{\mc{T}\}$, also known as the Euler-Hill frame, the dynamic equations become
\begin{align}
        \dot{r} &= v \,, \nonumber\\
        \dot{v} &= -S(\alpha)\,r -S^{2}(\omega)\,r - 2\,S(\omega)\,v  \nonumber\\
                &\hspace{11pt}-\mu \left( \frac{r+r_{t}}{\|r+r_{t}\|^3} - \frac{r_{t}}{\|r_{t}\|^3} \right) + u_{c} \,,     
    \label{eq:relative_orbital_dynamics_local}
\end{align}
where ${v\in\mbb{R}^{3}}$ is the relative velocity, ${\omega,\,\alpha\in\mbb{R}^{3}}$ are the target's local frame angular velocity and acceleration, respectively, and $S(\cdot)$ is the skew-symmetric matrix operator, such that ${\forall x, y \in \mbb{R}^{3},~ S(x)y = x\times y}$. 
Fig.\;\ref{fig:relative_navigation_diagram} shows a diagram of the problem at hand.
For a more detailed explanation of the equations of motion, see \cite[Ch. 4]{alfriend_formation_2009}.

\begin{figure}
    \centering
    \includegraphics[width=\columnwidth]{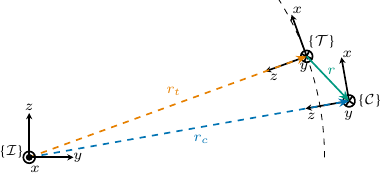}
    \caption{Relative navigation problem}
    \label{fig:relative_navigation_diagram}
\end{figure}

All quantities in \eqref{eq:relative_orbital_dynamics_local} are expressed in $\{\mc{T}\}$, reusing notation but contrasting with \eqref{eq:relative_orbital_dynamics_inertial}, expressed in  $\{\mc{I}\}$, and these are the equations of motion that will be used throughout this work.

Note, however, that the chaser's actuator input is known in the chaser's frame, $\{\mc{C}\}$. For ${r\ll r_{t}}$, the angular difference between the target's and the chaser's LVLH frames due to the curvature of the orbit is typically within the accuracy of attitude measurement and onboard calculation \cite[Ch. 7]{fehse_rendezvous_2003}. For example, with both spacecraft in the same LEO orbit, at a ${10^{3}\,\si{\kilo\meter}}$ altitude, and ${r\sim 10^{3}\,\si{\meter}}$, the angular difference between the two frames is approximately ${10^{-4}\,\si{\radian}}$. Since this difference is negligible, we will use $u_{c}$ expressed in $\{\mc{C}\}$ as if it were expressed in $\{\mc{T}\}$. 
In addition, $\alpha$, $\omega$, and $r_{t}$ are assumed to be obtained from ground stations. 

We assume that the chaser measures the bearing and range to the target, given by
\begin{equation}
    y_1 = -\frac{r}{\|r\|}, \quad y_2 = \|r\|.
    \label{eq:measurements}
\end{equation}
Note that, since the bearing measurement is done by the chaser, $y_1$ is the symmetric of the normalized relative position vector, again assuming that the angular difference between $\{\mc{C}\}$ and $\{\mc{T}\}$ is negligible.

\section{EQUIVARIANT FORMULATION}
\label{sec:equivariant_formulation}

\subsection{Group Symmetry}

For this problem, we propose using the Lie group ${\mr{G} = (\mr{SO}(3) \times \mr{MR}(1)) \ltimes \mr{R}^{3}}$, with Lie-algebra ${\mf{g} = (\mf{so}(3)\times\mf{mr}(1))\ltimes\mf{r}^3}$. 
This group, also known as the similarity group, $\mr{SIM}(3)$, is composed of $\mr{SO}(3)$, the special orthogonal group, $\mr{MR}(1)$, the multiplicative real group of positive scalars, and $\mr{R}^{3}$, the additive real group of 3-dimensional vectors.
An element of $\mr{G}$ is denoted by ${X = (R,s,z)}$,  where ${R\in\mr{SO}(3)}$ is a rotation matrix, ${s\in\mr{MR}(1)}$ is a positive real scalar, and ${z\in\mr{R}^{3}}$ is a 3-dimensional real vector. 
Its inverse is ${X^{-1} = (R,s,z)^{-1} = \big(R^\T, \frac{1}{s}, -\frac{1}{s}R^\T z\big)}$ and the group identity is ${\mr{id} = (I, 1, 0)}$.
For any two elements ${X_2, X_1 \in \mr{G}}$, the group product is 
\begin{align}
        X_2\cdot X_1 &= (R_2,s_2,z_2)\cdot(R_1,s_1,z_1) \nonumber\\
                     &= (R_2 R_1, s_2 s_1, z_2 + s_2 R_2 z_1).
\label{eq:group_product}
\end{align}
%
For an element of the Lie-algebra ${(\Omega,\tau,\eta)\in\mf{g}}$, the adjoint $\mr{Ad}_{X}(\Omega,\tau,\eta)$ is
\begin{equation}
    \mr{Ad}_{X}(\Omega,\tau,\eta) = \left(R\Omega R^\T, \tau, \frac{1}{s}R^{\T}\left((\Omega + \tau I)z +\eta\right)\right).
\end{equation}
This group is chosen for its compatibility with the assumed measurements of range and bearing, as will become clearer below, and because future work will pursue a bearing-only approach.

\subsection{System Equivariance}

We define the system state ${\xi = (r,v) \in \mc{M}}$, the input ${u = (\alpha, \omega, r_{t}, u_{c}, u_{v}, u_{w}) \in \mc{U}}$, and the output ${y = (y_1, y_2) \in \mc{Y}}$, where 
\begin{align}
    \mc{M} &= \mbb{R}^{3}\times\mbb{R}^{3},\\
    \mc{U} &= \mbb{R}^{3}\times\mbb{R}^{3}\times\mbb{R}^{3}\times\mbb{R}^{3}\times\mbb{R}^{3}\times\mbb{R},\\
    \mc{Y} &= \mc{S}^{2}\times \mbb{R}_{+},
\end{align}
are the state, input, and output manifolds, respectively, and $\mc{S}^2$ is the unit 2-sphere. Note that we have included ${u_{v}\in\mbb{R}^{3}}$ and ${u_{w}\in\mbb{R}}$, which are virtual inputs required to prove equivariance of the system \cite{mahony_observer_2022}. The extended version of system \eqref{eq:relative_orbital_dynamics_local}, including the virtual inputs, ${\dot{\xi}=f(\xi, u)}$, is
\begin{align}
        \dot{r} &= v + u_{v} \,, \nonumber\\
        \dot{v} &= -S(\alpha)\,r -S^{2}(\omega)\,r - 2\,S(\omega)\,(v + u_{v}) \nonumber\\
                &\hspace{11pt}-\mu \left( \frac{r+r_t}{\|r+r_t\|^3} - \frac{r_t}{\|r_t\|^3} \right)u_{w} + u_c \,.   
    \label{eq:relative_orbital_dynamics_local_extended}
\end{align}
The original system \eqref{eq:relative_orbital_dynamics_local} can be recovered from \eqref{eq:relative_orbital_dynamics_local_extended} by setting $u_{v}=0$ and $u_{w}=1$. For compactness, we will use the notation ${\dot{r}=f_r(\xi, u)}$ and ${\dot{v}=f_v(\xi, u)}$.

A system is said to be equivariant if it satisfies the condition
\begin{equation}
    D \phi(X, \xi)[f(\xi,u)] = f(\phi(X, \xi), \psi(X, u)),
    \label{eq:equivariant_system_condition}
\end{equation}
where ${\phi: \mr{G}\times\mc{M} \to \mc{M}}$ and ${\psi: \mr{G}\times\mc{U} \to \mc{U}}$ are right transitive actions that describe how the group acts on the state and input manifolds, respectively.
The notation $\phi_{X}(\xi)$ will be used to indicate that the argument $X$ is fixed and $\phi$ is evaluated with respect to $\xi$ and vice-versa for $\phi_{\xi}(X)$. The same logic applies to the input action $\psi$.

\begin{theorem}
    The right transitive state and input actions
    \begin{equation} 
        \phi\big((R,s,z), (r,v)\big) = \left(\frac{1}{s}R^\T r,\frac{1}{s}R^\T (v-z)\right), 
        \label{eq:phi_state_action}
    \end{equation}
    \begin{align}
        &\psi\left((R,s,z), (\alpha, \omega,r_{t},u_{c},u_{v},u_{w})\right) = \nonumber\\
        &\left(R^\T \alpha, R^\T \omega, \frac{1}{s}R^\T r_t, \frac{1}{s}R^\T u_c, \frac{1}{s^{3}} u_{w},\frac{1}{s}R^\T (u_v + z)\right)
    \label{eq:psi_input_action}
    \end{align}        
    satisfy condition \eqref{eq:equivariant_system_condition}. Therefore, the system \eqref{eq:relative_orbital_dynamics_local_extended} is equivariant.
\label{theo:system_equivariance}
\end{theorem}

\begin{proof}
    To prove equivariance of the system, we develop \eqref{eq:equivariant_system_condition} with the proposed state action $\phi$ and input action $\psi$. This leads to
    \begin{align}
        &D \phi(X, \xi)[f(\xi,u)]  \nonumber\\
        &\hspace{6pt}=\bigg( \frac{1}{s}R^\T(v + u_{v}), -\frac{1}{s}R^\T S(\alpha)r -\frac{1}{s}R^\T S^{2}(\omega) r \nonumber\\
        &\hspace{30pt}-2 \frac{1}{s}R^\T S(\omega) (v + u_{v}) - \mu \frac{1}{s}R^\T \frac{r+r_{t}}{\|r+r_{t}\|^3} u_{w} \nonumber\\
        &\hspace{30pt}+ \mu \frac{1}{s}R^\T \frac{r_{t}}{\|r_{t}\|^3} u_{w} + \frac{1}{s}R^\T u_{c} \bigg) \nonumber\\
        &\hspace{6pt}=f(\phi(X, \xi), \psi(X, u)).
    \label{eq:system_equivariance_proof}
    \end{align}
\end{proof}
    
\begin{remark}
    A group action is called transitive if it is surjective, i.e. ${\forall \xi, \mring{\xi} \in\mc{M},\; \exists X\in\mr{G}: \phi_{X}(\mring{\xi})=\xi}$. However, for ${\mring{\xi}=(0,\mring{v})}$, this property does not hold.
    Therefore, the group action $\phi_{X}$ is not transitive on the entire ${\mbb{R}^{3}\times \mbb{R}^{3}}$, but rather on ${\mbb{R}^{3}\setminus\{0\}\times \mbb{R}^{3}}$.
\label{rem:transitive_action_set}
\end{remark}

\subsection{Output Equivariance}
\label{subsec:output_equivariance}

The measurement function is ${y = h(\xi) = \left(-\frac{r}{\|r\|}, \|r\|\right)}$. We define the right output action ${\rho : \mr{G} \times \mc{Y} \to \mc{Y}}$ as
\begin{equation}
    \rho\big((R,s,z),(y_1, y_2)\big) = \left(R^\T y_1, \frac{1}{s}y_2\right) .
    \label{eq:right_ouput_action}
\end{equation}
Since ${\rho(X,y)=h(\phi(X,\xi))}$, we conclude that the output is equivariant \cite{mahony_observer_2022}.

\subsection{Equivariant Lift}

Next, we derive an equivariant lift of the system. The lift  ${\Lambda: \mc{M} \times \mc{U} \to \mf{g}}$ is a map that lifts the dynamics of the system from the state manifold to the Lie-algebra of the symmetry group \cite{mahony_equivariant_2021, mahony_observer_2022}.
An equivariant lift must satisfy the conditions
\begin{gather}
    D_{|_{X=\mr{id}}}\phi_{\xi}(X)[\Lambda(\xi, u)] = f(\xi, u),\label{eq:lift_condition_1}\\
    \mr{Ad}_{X^{-1}}(\Lambda(\xi, u)) = \Lambda(\phi_{X}(\xi), \psi_{X}(u)).\label{eq:lift_condition_2}
\end{gather}
For brevity, we will at times write the lift (and each of its components) without the arguments, i.e., ${\Lambda=\Lambda(\xi, u)}$.
\begin{theorem}
The lift $\Lambda = (\Lambda_R, \Lambda_s, \Lambda_z)$ given by
\begin{subequations}
\label{eq:equivariant_lift}
\begin{align}
    \Lambda_R &= -S\bigg( \frac{r \times (v + u_{v})}{\|r\|^{2}} \bigg) ,\label{eq:equivariant_lift_R}\\ 
    \Lambda_s &= -\frac{r^\T (v + u_{v})}{\|r\|^{2}} , \label{eq:equivariant_lift_s}\\
    \Lambda_z &= S\bigg( \frac{r \times (v + u_{v})}{\|r\|^{2}} \bigg) v + \frac{r^\T (v + u_{v})}{\|r\|^{2}} v \nonumber\\
              & \hspace{10pt}+ S(\alpha) r + S^{2}(\omega) r + 2 S(\omega) (v + u_{v}) \nonumber\\
              & \hspace{10pt}+ \mu \bigg( \frac{r+r_t}{\|r+r_t\|^3} - \frac{r_t}{\|r_t\|^3} \bigg)u_{w} - u_c \label{eq:equivariant_lift_z}
\end{align}
\end{subequations}
is equivariant.
\end{theorem}
\begin{proof}
We need to analyze conditions \eqref{eq:lift_condition_1} and \eqref{eq:lift_condition_2} to verify that the lift is, in fact, equivariant. Let us start by developing condition \eqref{eq:lift_condition_1}, which results in
\begin{align}
    &D_{|_{X=\mr{id}}}\phi_{\xi}(X)[\Lambda(\xi, u)] \nonumber\\
    &\hspace{8pt}= (-\Lambda_R r - \Lambda_s r, -\Lambda_z -\Lambda_R v - \Lambda_s v) \nonumber\\
    &\hspace{8pt}= \Bigg( S\bigg( \frac{r \times (v + u_{v})}{\|r\|^{2}} \bigg)r +\frac{r^\T (v + u_{v})}{\|r\|^{2}} r, \nonumber\\
    &\hspace{30pt}   -S\bigg( \frac{r \times (v + u_{v})}{\|r\|^{2}} \bigg) v - \frac{r^\T (v + u_{v})}{\|r\|^{2}} v \nonumber\\
    &\hspace{30pt}- S(\alpha) r - S^{2}(\omega) r - 2 S(\omega) (v + u_{v}) \nonumber\\
    &\hspace{30pt}- \mu \bigg( \frac{r+r_t}{\|r+r_t\|^3} - \frac{r_t}{\|r_t\|^3} \bigg)u_{w} + u_c  \nonumber\\
    &\hspace{30pt}+ S\bigg( \frac{r \times (v + u_{v})}{\|r\|^{2}} \bigg)v +\frac{r^\T (v + u_{v})}{\|r\|^{2}} v \Bigg) \nonumber\\ 
    &\hspace{8pt}=\bigg(v + u_{v}, - S(\alpha) r - S^{2}(\omega) r - 2 S(\omega) (v + u_{v}) \nonumber\\
    &\hspace{60pt}- \mu \bigg( \frac{r+r_t}{\|r+r_t\|^3} - \frac{r_t}{\|r_t\|^3} \bigg)u_{w} + u_c \bigg) \nonumber\\
    &\hspace{8pt}=f(\xi,u).
\end{align}
As such, condition \eqref{eq:lift_condition_1} is satisfied.
Next, we analyze condition \eqref{eq:lift_condition_2}. The left-hand side is
\begin{equation}
    \mr{Ad}_{X^{-1}}(\Lambda) = \left(R^\T \Lambda_R R, \Lambda_s, \frac{1}{s} R^\T (\Lambda_R + \Lambda_s I) z + \frac{1}{s}R^\T \Lambda_z\right)
    \label{eq:lift_condition_2_lhs}
\end{equation} 
We analyze each component of the right-hand side of \eqref{eq:lift_condition_2} separately. For $\Lambda_R$, we have
\begin{align}
   &\Lambda_R(\phi_{X}(\xi), \psi_{X}(u)) \nonumber\\
   &\hspace{2pt}=-S\left( \frac{\frac{1}{s}R^\T r \times \left(\frac{1}{s}R^\T(v-z) + \frac{1}{s}R^\T(u_{v}+z)\right)}{\|\frac{1}{s}R^\T r\|^{2}} \right) \nonumber\\
   &\hspace{2pt}=- S\bigg( R^\T \frac{r \times (v + u_{v})}{\|r\|^{2}} \bigg) \nonumber\\ 
   &\hspace{2pt}=-R^\T S\left( \frac{r \times (v + u_{v})}{\|r\|^{2}} \right) R \nonumber\\
   &\hspace{2pt}=R^\T \Lambda_R R .
\end{align}
For the next component, $\Lambda_s$, we find that
\begin{align}
   &\Lambda_s(\phi_{X}(\xi), \psi_{X}(u)) \nonumber\\
   &\hspace{2pt}= -\frac{\left(\frac{1}{s}R^\T r\right)^\T \left(\frac{1}{s}R^\T(v-z) + \frac{1}{s}R^\T(u_{v}+z)\right)}{\|\frac{1}{s}R^\T r\|^{2}} \nonumber\\
   &\hspace{2pt}=-\frac{r^\T (v + u_{v})}{\|r\|^{2}} \nonumber\\
   &\hspace{2pt}= \Lambda_s .
\end{align}
For the third component, note that $\Lambda_z = -(\Lambda_R + \Lambda_s I)v - f_v$. Then, by substituting the expression, it follows
\begin{align}
    &\Lambda_z(\phi_{X}(\xi), \psi_{X}(u)) \nonumber\\
    &= -\Lambda_R(\phi_{X}(\xi), \psi_{X}(u)) \frac{1}{s}R^\T(v-z) \nonumber\\
    &\hspace{11pt} -\Lambda_s(\phi_{X}(\xi), \psi_{X}(u)) \frac{1}{s}R^\T(v-z) - f_v(\phi_{X}(\xi), \psi_{X}(u)) \nonumber\\
    &= -(R^\T\Lambda_R R + \Lambda_s I) \frac{1}{s}R^\T(v-z) - \frac{1}{s}R^\T f_v \nonumber\\ 
    &= \frac{1}{s}R^\T(\Lambda_R + \Lambda_s I) z -\frac{1}{s}R^\T(\Lambda_R + \Lambda_s I) v - \frac{1}{s}R^\T f_v \nonumber\\ 
    &= \frac{1}{s} R^\T \big((\Lambda_R + \Lambda_s I) z + \Lambda_z \big) .
\end{align}

Thus, conditions \eqref{eq:lift_condition_1} and \eqref{eq:lift_condition_2} are both met and we conclude that the proposed lift \eqref{eq:equivariant_lift} is an equivariant lift of the system.
\end{proof}

One can gain intuition about the lift by noting that $\Lambda_{s}$ is related to the time-derivative of the inverse of the range and that ${(r\times v)/\|r\|^2}$ is the bearing's angular velocity, indicating that $\Lambda_{R}$ is related to the rotational kinematics.

\subsection{Filter Derivation}
\label{subsec:filter_derivation}

The true state of the system on the group is denoted by $X=(R,s,z)$ and the filter state is $\hat{X} = (\hat{R}, \hat{s}, \hat{z})$. The filter dynamics, on the group, are given by 
\begin{equation}
    \begin{split}
        \dot{\hat{R}} &= \hat{R}\Lambda_R(\phi(\hat{X},\mring{\xi}), u) + \Delta_R \hat{R}\\
        \dot{\hat{s}} &= \hat{s}\Lambda_s(\phi(\hat{X},\mring{\xi}), u) + \delta_s \hat{s}\\
        \dot{\hat{z}} &= \hat{s}\hat{R}\Lambda_z(\phi(\hat{X},\mring{\xi}), u) + \Delta_R \hat{z} + \delta_s \hat{z} + \delta_z ,
    \end{split}
    \label{eq:filter_dynamics}
\end{equation}
where $\Delta = (\Delta_R, \delta_s, \delta_z) \in \mf{g}$ are correction terms that will be defined shortly. The estimate on the state manifold is 
\begin{equation}
    \hat{\xi} = \left(\hat{r},\hat{v}\right) = \phi(\hat{X}, \mring{\xi}) = \left( \frac{1}{\hat{s}} \hat{R}^{\T} \mring{r}, \frac{1}{\hat{s}} \hat{R}^{\T} \left(\mring{v} - \hat{z}\right) \right),
    \label{eq:manifold_state_estimates}
\end{equation}
where ${\mring{\xi}=(\mring{r},\mring{v})}$ is the origin of the coordinate system on the manifold, which is required to map the filter’s estimates from the group to the state manifold.
Since $\phi_{X}$ is not transitive on the entire state manifold, as pointed out in Remark\;\ref{rem:transitive_action_set}, the origin cannot be ${\mring{\xi} = (0,0)}$, as one might initially expect. As such, let the origin be ${(\mring{r},\mring{v}) = (\gamma\mr{e}_{3}, 0)}$, where ${\mr{e}_{3}=[0\;0\;1]^\T}$ is the third canonical basis vector of $\mbb{R}^{3}$ and ${\gamma\in\mbb{R}\setminus\{0\}}$.

To calculate the correction terms, we need to characterize the filter error. We begin by defining the filter error on the group, ${E=(E_R, E_s, E_z)}$, as
\begin{equation}
    E = (R, s, z) \cdot (\hat{R}, \hat{s}, \hat{z})^{-1} = \left(R\hat{R}^\T, \frac{s}{\hat{s}}, z - \frac{s}{\hat{s}} R\hat{R}^\T \hat{z}\right).
\label{eq:error_group}
\end{equation}
Then, the error on the state manifold, ${e= (e_r, e_v)}$, is obtained using $E$ and the state action $\phi$ as
\begin{equation}
    e = \phi(E,\mring{\xi})= \left(\frac{1}{E_s}E_R^{\T}\mring{r}, \frac{1}{E_s}E_R^{\T}(\mring{v} - E_z)\right).
    \label{eq:error_manifold}
\end{equation}

We also need to define the error on the tangent space at the origin $\mring{\xi}$. We fix a local coordinate chart ${\vartheta:\mc{N}_{\mring{\xi}}\to \mc{T}_{\mring{\xi}}\mc{M}}$, where ${\mc{N}_{\mring{\xi}}\subset\mc{M}}$ is a neighborhood of the origin and ${\mc{T}_{\mring{\xi}}\mc{M}\equiv\mbb{R}^{6}}$ is the tangent space at the origin.
The local coordinate chart is chosen to be the normal coordinates on $\mc{M}$, obtained by projecting exponentials from $\mr{G}$ onto the state manifold, around the origin, using the action $\phi$.
That is, with ${\varepsilon\in\mbb{R}^{6}}$ being the error on the tangent space, we have that ${\varepsilon = \vartheta(e)}$ and $e = \vartheta^{-1}(\varepsilon) = \phi(\mr{exp}(\varepsilon^\wedge), \mring{\xi})$, where ${\varepsilon^\wedge\in\mf{g}}$. The local coordinate chart is given by
\begin{equation}
    \vartheta(e) = \bigg( \bmat{\mr{e}_2^\T\\-\mr{e}_1^\T} \frac{e_{r}}{\|e_{r}\|}, \mr{log}\bigg(\frac{\|\mring{r}\|}{\|e_{r}\|}\bigg), -e_{v} \bigg),
\end{equation}
where ${\mr{e}_{1}=[1\;0\;0]^\T}$ and ${\mr{e}_{2}=[0\;1\;0]^\T}$ are the first and second canonical basis vectors of ${\mbb{R}^{3}}$, respectively.

To design the EqF, we need to characterize the error dynamics of the system. The time derivative of the error on the state manifold is given by
\begin{equation}
    \dot{e} = D\phi_{e}(\Lambda(e,\mring{u}) - \Lambda(\mring{\xi}, \mring{u})) - D\phi_{e}\Delta ,
\end{equation}
where ${\mring{u} = \psi(\hat{X}^{-1}, u)}$. Assuming that ${\forall t,~e(t) \in \mc{N}_{\mring{\xi}}}$ and knowing that ${\varepsilon = \vartheta(e)}$, the dynamics of the error on the tangent space are
\begin{equation}
    \dot{\varepsilon} = D\vartheta\cdot D\phi_{e}(\Lambda(e,\mring{u}) - \Lambda(\mring{\xi}, \mring{u})) - D\vartheta \cdot D \phi_{e}\Delta .
\end{equation}

The pre-observer dynamics of the error on the tangent space are linearized around ${\varepsilon=0}$, resulting in ${\dot{\varepsilon} = \mring{A}_t \varepsilon + \mc{O}(\|\varepsilon\|^2)}$. The matrix $\mring{A}_t$ is given by
\begin{equation}
    \mring{A}_t = D_{|_{e=\mring{\xi}}}\vartheta(e) D_{|_{E=\mr{id}}}\phi_{\mring{\xi}}(E) D_{|_{e=\mring{\xi}}} \Lambda(e,\mring{u}) D_{|_{\varepsilon=0}} \vartheta^{-1}(\varepsilon) .
    \label{eq:matrix_A_equation}
\end{equation}
The output residual ${\Tilde{y} = y-\hat{y}}$ is also linearized around ${\varepsilon=0}$. Since the output is equivariant, as seen in \ref{subsec:output_equivariance}, we have that ${\Tilde{y} = C_t \varepsilon + \mc{O}(\|\varepsilon\|^3)}$ and
\begin{equation}
    C_t \varepsilon = \frac{1}{2} \big( D_{|_{E=\mr{id}}}\rho(E,y) + D_{|_{E=\mr{id}}}\rho(E,\hat{y}) \big) \cdot \mr{Ad}_{\hat{X}^{-1}}(\varepsilon^\wedge).
    \label{eq:matrix_C_equation}
\end{equation}
The explicit expressions for the matrices $\mring{A}_t$ and $C_t$ are presented in the Appendix. 

The correction terms also depend on the Riccati state, ${\Sigma \in \mbb{S}^{6}_{+}}$, which evolves according to
\begin{equation}
    \dot{\Sigma} = \mring{A}_t \Sigma + \Sigma \mring{A}_t^\T + M_t - \Sigma C_{t}^{\T} N_{t}^{-1} C_{t} \Sigma ,
    \label{eq:riccati_state_dynamics}
\end{equation}
with ${\Sigma(0) = \Sigma_{0}}$ and ${\Sigma_{0} \in \mbb{S}^{6}_{+}}$. The matrices ${M_t \in \mbb{S}^{6}_{+}}$ and ${N_t \in \mbb{S}^{4}_{+}}$ are the state and output gain matrices, respectively, and $\mbb{S}^{n}_{+}$ is the set positive definite ${n\times n}$ matrices.

Finally, the correction terms are calculated according to
\begin{equation}
    \Delta = D_{|_{E=\mr{id}}}\phi_{\mring{\xi}}(E)^{\dagger} \cdot D \vartheta^{-1} \big[\Sigma C_{t}^{\T} N_{t}^{-1}(y - \hat{y})\big],
\end{equation}
where $(\cdot)^{\dagger}$ denotes a right inverse and ${\hat{y}=h(\hat{\xi})}$.

\subsection{Convergence Analysis}

With range and bearing measurements, the problem of estimating the relative position and velocity between the two spacecraft is observable. Given this observability condition, the EqF can be shown to be locally exponentially stable. We refer the reader to \cite{van_goor_equivariant_filter_2023, serrano_equivariant_attitude_2026} for further details.

\section{SIMULATIONS}
\label{sec:simulations}

To verify the proposed Equivariant Filter's performance, we conduct Monte Carlo simulations using MATLAB.

\subsection{Simulation setup}
\label{subsec:simulation_setup}

We simulate two scenarios: a circumnavigation orbit and an r-bar approach, depicted in Fig.\;\ref{fig:simulation_trajectories} in the LVLH frame. In both cases, the target is in a circular LEO orbit at an altitude of 1000\,\si{\kilo\meter}, with orbital period ${T\approx105\,\si{\minute}}$, and $\omega$, $\alpha$ and $r_{t}$ are assumed known from ground observation, as explained in Section\;\ref{sec:problem_statement}. For a circular orbit, $\omega$ and $r_{t}$ are constant, as they are expressed in $\{\mc{T}\}$, and ${\alpha = 0}$. Though the method is not restricted to circular orbits, they are chosen for simplicity. Each trajectory is simulated 100 times.

The Equivariant Filter is implemented in a prediction-update formulation. The prediction step is computed at every 0.01\,\si{\second}, while the update step is performed at every 1\,\si{\second}, as the measurements are assumed to be acquired synchronously at 1\,\si{\hertz}. 
Measurement noise is considered such that the bearing measurement has an angular deviation of $\theta$ from the true value, with $\theta \sim \mathpzc{N}(0, 0.01^2)$, and the range measurement is within 1\,\si{\percent} of the true range.
The initial filter estimates in the group are set to the identity, i.e., ${(\hat{R}, \hat{s}, \hat{z}) = (I,1,0)}$, with gain matrices ${M_{t}=0.1 I}$ and ${N_{t}=\mr{diag}([1\,1\,1\,10^{4}])}$. The initial value of the Riccati state is ${\Sigma_{0}=I}$ and the manifold origin is chosen to be ${\mring{\xi}=(10^{2}\mr{e}_{3},0)}$.

\begin{figure}[htbp]
    \centering
    \begin{subfigure}{\columnwidth}
        \centering
        \includegraphics[width=\columnwidth]{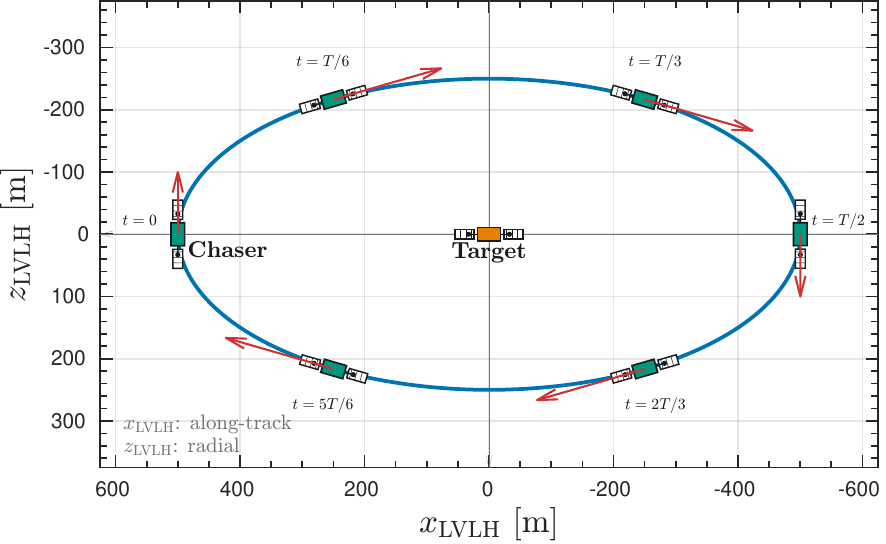}
        \caption{Circumnavigation orbit.}
        \label{fig:relative_circumnavigation_LVLH}
    \end{subfigure}
    \begin{subfigure}{\columnwidth}
        \centering
        \includegraphics[width=\columnwidth]{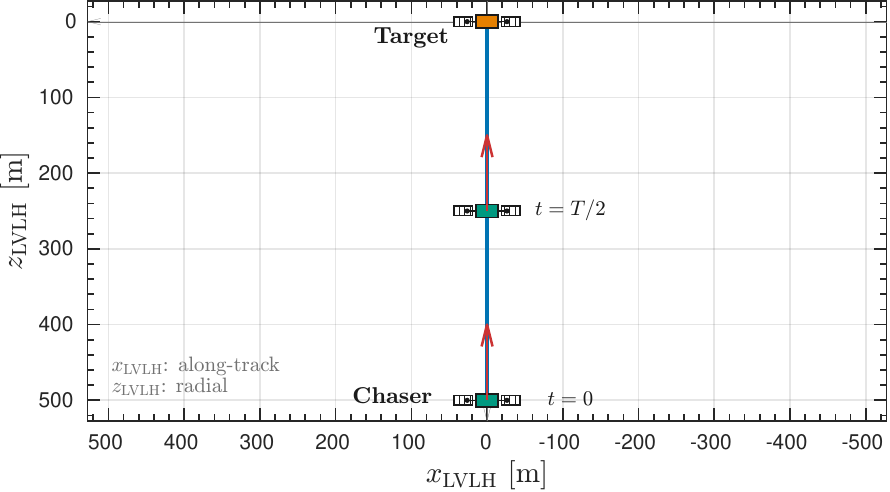}
        \caption{R-bar approach.}
        \label{fig:relative_rbar_approach_LVLH}
    \end{subfigure}
    \caption{Simulated proximity operation scenarios.}
    \label{fig:simulation_trajectories}
\end{figure}

\subsection{Circumnavigation}
\label{subsec:circumnavigation}

A circumnavigation orbit is an elliptical trajectory around the target that does not require active thrusting.
It is often employed to keep the chaser in a passively safe, thrust-free trajectory looping around the target for inspections or mapping \cite{woffinden_angles-only_2008}. In the simulations, the ellipse that the chaser describes, in the $xz$-plane of the local frame, i.e. the orbital plane, has a semi-major axis of 500\,\si{\meter} and a semi-minor axis of 250\,\si{\meter}. The chaser takes one orbital period to circumnavigate the target.

In Fig.\;\ref{fig:circumnavigation_approach_manifold_estimates}, we show the filter estimates mapped onto the state manifold, i.e., $\hat{r}$ and $\hat{v}$, in solid lines, along with the true values of the relative position and velocity, in dashed lines, for a representative Monte Carlo run. The norm of the manifold error, averaged across the Monte Carlo simulations, is plotted in Fig.\;\ref{fig:circumnavigation_manifold_errors_norm_log}, with the origin subtracted, since ${e\to\mring{\xi}}$ as the filter converges. The shaded areas represent the entire span of values in all simulations. We present the results up to one quarter of the orbit, as the filter converges within this time span.
%
From the plots, one concludes that the filter converges in less than one-tenth of the orbital period. In the time interval $[0.1,\,0.25]\,T$, the mean of the position error norm is approximately ${1.45\cdot10^{-1}\,\si{\meter}}$. In relative terms, the estimate of the position is on average within \SI{0.05}{\percent} of the true value. The mean velocity error norm is ${1.24\cdot10^{-3}\,\si{\meter\per\second}}$, corresponding to a mean relative error of \SI{0.28}{\percent}.



\begin{figure}
    \centering
    \begin{subfigure}{\columnwidth}
        \centering
        \includegraphics[width=\columnwidth]{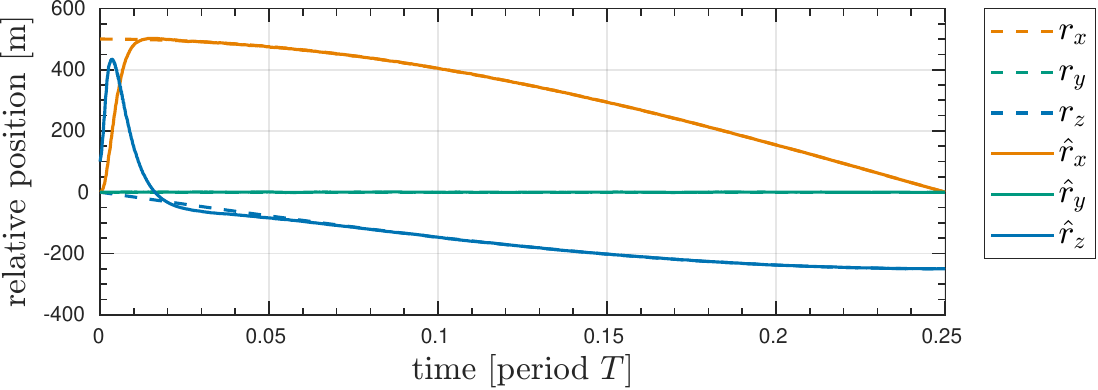}
        \caption{Relative position.}
        \label{fig:circumnavigation_manifold_position_estimate}
    \end{subfigure}
    \begin{subfigure}{\columnwidth}
        \centering
        \includegraphics[width=\columnwidth]{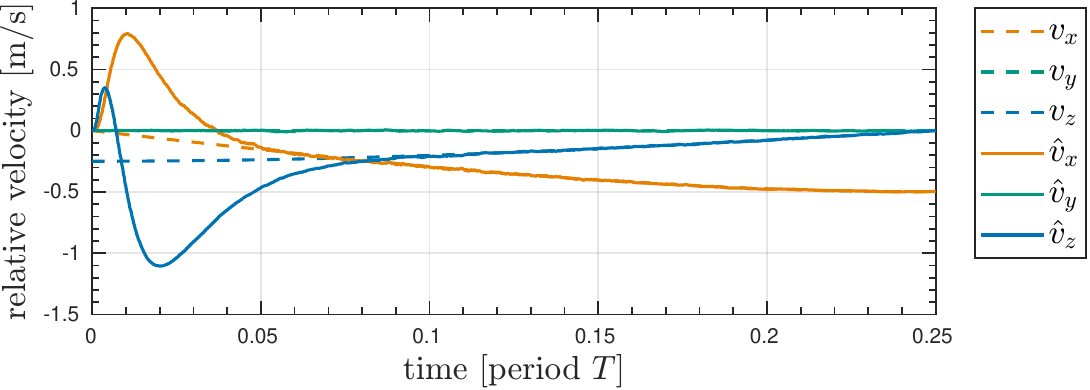}
        \caption{Relative velocity.}
        \label{fig:circumnavigation_manifold_velocity_estimate}
    \end{subfigure}
    \caption{Circumnavigation: estimated (solid) and true (dashed) relative states.}
    \label{fig:circumnavigation_approach_manifold_estimates}
\end{figure}

\begin{figure}
    \centering
    \includegraphics[width=\columnwidth]{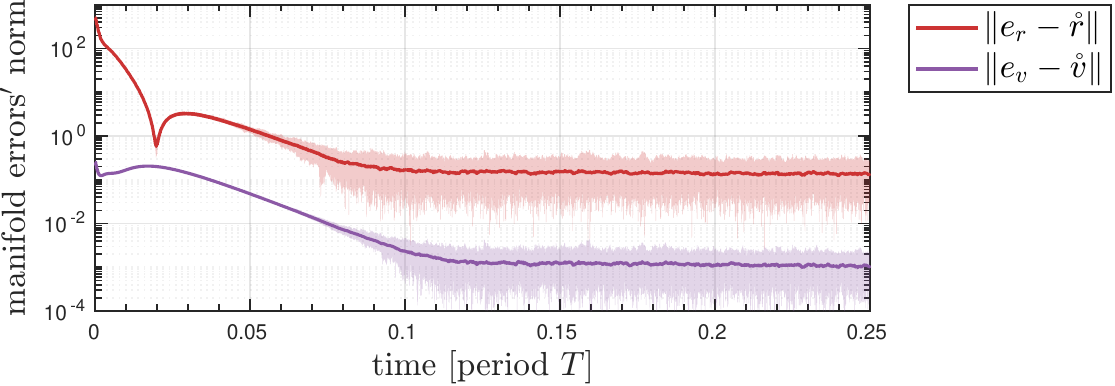}
    \caption{Circumnavigation: averaged manifold errors' norm.}
    \label{fig:circumnavigation_manifold_errors_norm_log}
\end{figure}

\subsection{R-bar approach}
\label{subsec:rbar_approach}
The second scenario is a straight line r-bar approach, from below the target, with constant velocity and continuous actuation, as described in \cite{fehse_rendezvous_2003}. The chaser starts at a distance of 500\,\si{\meter} from the target and takes one orbital period to rendezvous. Again, we show the results up to one quarter of the orbital period, since the filter converges within this time interval.
The filter estimates on the state manifold are shown in Fig.\;\ref{fig:rbar_approach_manifold_estimates}, for a representative simulation, and the norm of the manifold error across the Monte Carlo simulations is presented in Fig.\;\ref{fig:rbar_approach_manifold_errors_norm_log}. 
%
On average, in the time interval $[0.1,\,0.25]\,T$, the position error norm is approximately ${1.48\cdot10^{-1}\,\si{\meter}}$, which corresponds to \SI{0.04}{\percent} of the true relative position. The mean velocity error norm is ${1.19\cdot10^{-3}\,\si{\meter\per\second}}$, i.e. \SI{1.50}{\percent} of the true value.



\begin{figure}
    \centering
    \begin{subfigure}{\columnwidth}
        \centering
        \includegraphics[width=1.01\columnwidth]{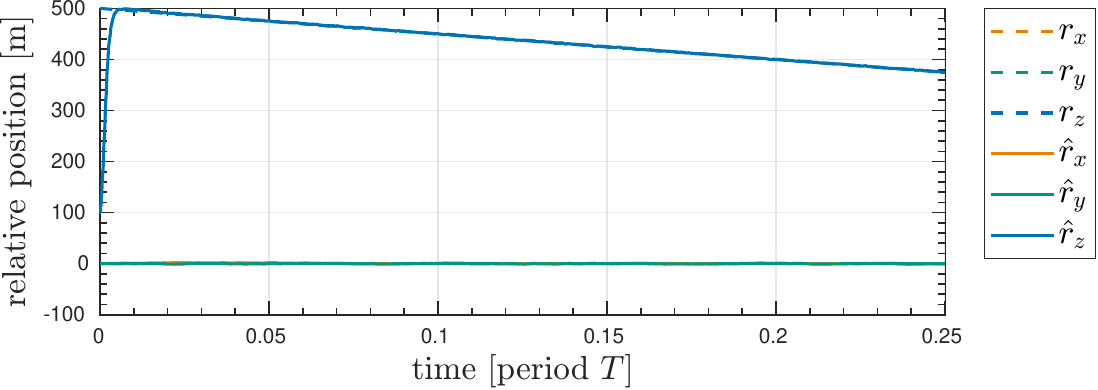}
        \caption{Relative position}
        \label{fig:rbar_approach_manifold_position_estimate}
    \end{subfigure}
    \begin{subfigure}{\columnwidth}
        \centering
        \includegraphics[width=\columnwidth]{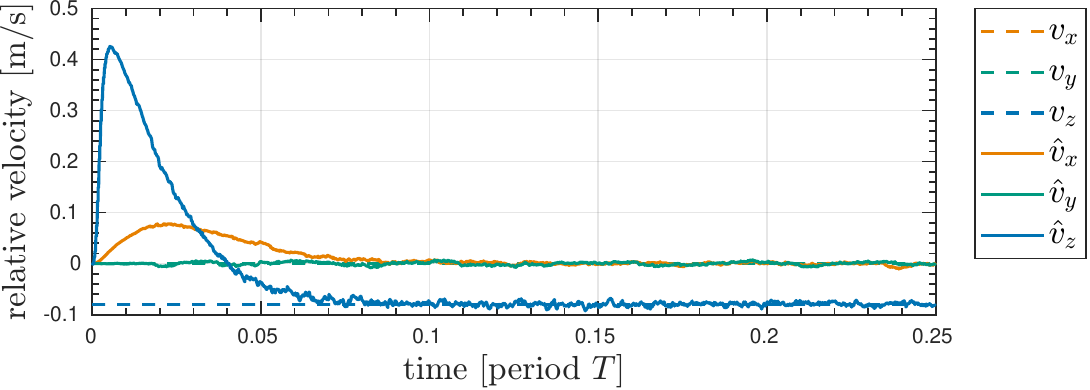}
        \caption{Relative velocity}
        \label{fig:rbar_approach_manifold_velocity_estimate}
    \end{subfigure}
    \caption{R-bar approach: estimated (solid) and true (dashed) relative states.}
    \label{fig:rbar_approach_manifold_estimates}
\end{figure}

\begin{figure}[t]
    \centering
    \includegraphics[width=\columnwidth]{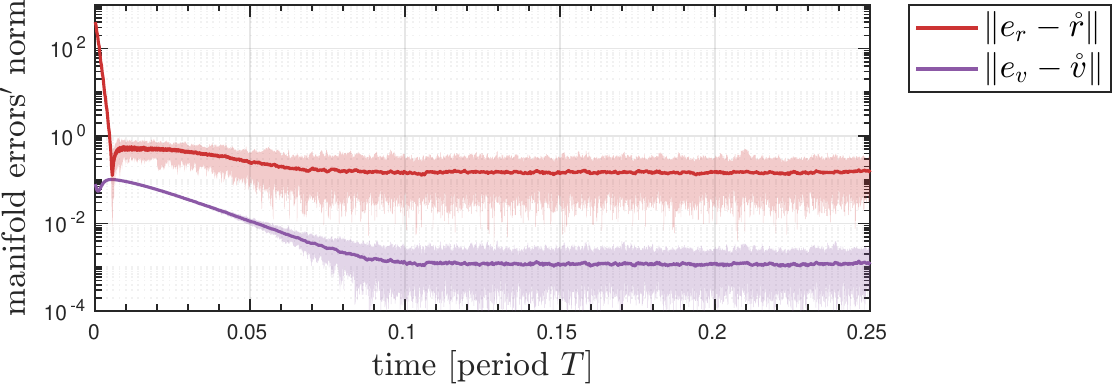}
    \caption{R-bar approach: averaged manifold errors' norm.}
    \label{fig:rbar_approach_manifold_errors_norm_log}
\end{figure}

\section{DISCUSSION}
\label{sec:discussion}

The EqF was applied to two proximity operation scenarios which are relevant for OOS/ADR missions. 
The results show that the filter is able to accurately estimate the relative position and velocity between the target and the chaser spacecraft, given noisy bearing and range measurements, at 1\,\si{\hertz} frequency. 
The filter converges in less than one-tenth of the orbit. 
Despite the difference in the trajectories, the same tuning and filter initial conditions are used for both scenarios, yielding the aforementioned results.
We also note that, while the system equations on the state manifold use a cartesian coordinates parametrization, the equations on the group are more akin to a polar coordinate parametrization, which is a consequence of the chosen symmetry group.

Nevertheless, there are still some limitations that need to be addressed. 
In these simulations, we do not take into account orbital disturbances, such as those related to residual atmosphere or geopotential anomaly.
As we do not consider what exact sensors are used, no sensor constraints, such as field-of-view, nor environmental perturbations, such as sunlight reflections, 
are examined.
The assumption of synchronous range and bearing measurements might not always be valid and analyzing only one measurement rate, though within the typical interval for onboard sensors, is limiting. 
Complete knowledge of $\omega$, $\alpha$ and $r_{t}$ is considered, while in reality there will be uncertainty associated with these quantities. 
The choice of the origin $\mring{\xi}$, particularly $\mring{r}$, has practical consequences beyond those explained in \ref{subsec:filter_derivation}. A disproportionate choice of $\mring{r}$, with respect to the distances at hand, especially at the initial time, may cause a large transient response and compromise numerical stability. 
There are terms dependent on $1/\|\hat{r}\|$ which may also cause computational issues as ${\|r\|\to0}$ during the final moments of rendezvous. Nonetheless, adequate measures can be taken from an implementation standpoint, to avoid numerical problems. For instance, the navigation system may use additional filtering algorithms running in parallel, both for usage at very small distances and to provide redundancy.
As such, additional testing regarding these issues is necessary.

Ultimately, despite further research being required, the results obtained are promising and validate the equivariant structure for the nonlinear relative orbital dynamics and the application of the filter to spacecraft relative navigation.

\section{CONCLUSIONS}
\label{sec:conclusions}

In this letter, we presented an equivariant structure for the nonlinear relative orbital dynamics, with range and bearing measurements, and designed an Equivariant Filter. Simulations of a circumnavigation orbit and an r-bar approach demonstrated the filter's performance. 
Future work will delve into the limitations regarding the measurements, different orbits and perturbations, and explore a bearing-only variant.


\addtolength{\textheight}{-9cm}   


\section*{APPENDIX}
\label{sec:appendix}

To improve readability, we decompose the matrix $\mring{A}_{t}$ into nine (different sized) components, as
\begin{equation}
    \mring{A}_t = 
    \begin{bmatrix}
        \mring{A}_{1,1} & \mring{A}_{1,2} & \mring{A}_{1,3} \\
        \mring{A}_{2,1} & \mring{A}_{2,2} & \mring{A}_{2,3} \\
        \mring{A}_{3,1} & \mring{A}_{3,2} & \mring{A}_{3,3} \\
    \end{bmatrix}.
\end{equation}
Each of the above components of the matrix is given by
\begin{align}
    \mring{A}_{1,1} &= \frac{1}{\|\mring{r}\|^2} \bmat{\mr{e}_1^\T\\\mr{e}_2^\T} S(\mring{u}_{v})S(\mring{r})\bmat{\mr{e}_1&\mr{e}_2} ,\\
    \mring{A}_{1,2} &= -\frac{1}{\|\mring{r}\|^2} \bmat{\mr{e}_1^\T\\\mr{e}_2^\T}S(\mring{r})\mring{u}_{v} ,\\
    \mring{A}_{1,3} &= \frac{1}{\|\mring{r}\|^2} \bmat{\mr{e}_1^\T\\\mr{e}_2^\T}S(\mring{r}) ,\\    
    \mring{A}_{2,1} &=-\frac{1}{\|\mring{r}\|^2} \mring{u}_{v}^\T S(\mring{r})\bmat{\mr{e}_1&\mr{e}_2} ,\\
    \mring{A}_{2,2} &=-\frac{1}{\|\mring{r}\|^2} \mring{r}^\T \mring{u}_{v} ,\\
    \mring{A}_{2,3} &=\frac{1}{\|\mring{r}\|^2} \mring{r}^\T ,\\
    \mring{A}_{3,1} &=\bigg(S(\mring{\alpha})+S^{2}(\mring{\omega}) +\mu \frac{1}{\|\mring{r}+\mring{r}_t\|^3}\mring{u}_w \nonumber\\
        &\hspace{16pt}-3\mu\frac{(\mring{r}+\mring{r}_t)(\mring{r}+\mring{r}_t)^\T}{\|\mring{r}+\mring{r}_t\|^5}\mring{u}_w\bigg) S(\mring{r})\bmat{\mr{e}_1&\mr{e}_2} ,\\
    \mring{A}_{3,2} &=-\bigg(S(\mring{\alpha})+S^{2}(\mring{\omega}) +\mu \frac{1}{\|\mring{r}+\mring{r}_t\|^3}\mring{u}_w \nonumber\\
        &\hspace{26pt}-3\mu\frac{(\mring{r}+\mring{r}_t)(\mring{r}+\mring{r}_t)^\T}{\|\mring{r}+\mring{r}_t\|^5}\mring{u}_w \bigg)\mring{r} ,\\
    \mring{A}_{3,3} &=S\left(-\frac{\mring{r}\times\mring{u}_{v}}{\|\mring{r}\|^2}\right) - \frac{\mring{r}^\T \mring{u}_{v}}{\|\mring{r}\|^2}I - 2S(\mring{u}_{\omega}) .
\end{align}

\noindent The matrix $C_t$ is given by
\begin{equation}
    C_t = \frac{1}{2}\begin{bmatrix}
        S(y_1 + \hat{y}_1)\hat{R}^\T \begin{bmatrix}
            \mr{e}_1 & \mr{e}_2
        \end{bmatrix} & 0 & 0\\
        0 & -(y_2 + \hat{y}_2) & 0
    \end{bmatrix}.
    \label{eq:matrix_C}
\end{equation}





\bibliographystyle{IEEEtran}
\bibliography{IEEEabrv,bibliography}

\end{document}